\documentclass[a4paper,UKenglish,cleveref, autoref, thm-restate]{lipics-v2019}

\renewcommand{\Pr}[1]{{\mathop{\text{Pr}}\left[#1\right]}}
\newcommand{\Ex}[1]{{\mathop{\mathbb{E}}\left[#1\right]}}

\usepackage{algorithm}
\usepackage{lineno}

\usepackage{pgfplots}
\pgfplotsset{compat=newest}

\nolinenumbers

\title{A Simple Algorithm for Minimum Cuts in Near-Linear Time} 

\titlerunning{Simple Min-Cut in Near-Linear Time}

\author{Nalin Bhardwaj}{University of California San Diego, USA}{nalinbhardwaj@nibnalin.me}{}{}

\author{Antonio J. Molina Lovett}{Department of Computer Science, Princeton University, USA}{antonio@amolina.ca}{}{}

\author{Bryce Sandlund}{Cheriton School of Computer Science, University of Waterloo, Canada}{bcsandlund@gmail.com}{}{}

\authorrunning{N. Bhardwaj, A. J. Molina Lovett, and B. Sandlund}

\Copyright{Nalin Bhardwaj, Antonio J. Molina Lovett, and Bryce Sandlund}



\category{}

\relatedversion{}

\supplement{}

\EventEditors{Susanne Albers}
\EventNoEds{1}
\EventLongTitle{17th Scandinavian Symposium and Workshops on Algorithm Theory (SWAT 2020)}
\EventShortTitle{SWAT 2020}
\EventAcronym{SWAT}
\EventYear{2020}
\EventDate{June 22--24, 2020}
\EventLocation{T\'{o}rshavn, Faroe Islands}
\EventLogo{}
\SeriesVolume{162}
\ArticleNo{30}

\begin{document}

\maketitle

\begin{abstract}
We consider the minimum cut problem in undirected, weighted graphs. We give a simple algorithm to find a minimum cut that $2$-respects (cuts two edges of) a spanning tree $T$ of a graph $G$. This procedure can be used in place of the complicated subroutine given in Karger's near-linear time minimum cut algorithm~\cite{Karger00}. We give a self-contained version of Karger's algorithm with the new procedure, which is easy to state and relatively simple to implement. It produces a minimum cut on an $m$-edge, $n$-vertex graph in $O(m \log^3 n)$ time with high probability, matching the complexity of Karger's approach.
\end{abstract}

\pagebreak

\section{Introduction}

The minimum cut problem on an undirected (weighted) graph $G$ asks for a vertex subset $S$ such that the total number (weight) of edges from $S$ to $V \setminus S$ is minimized. The minimum cut problem is a fundamental problem in graph optimization and has received vast attention by the research community across a number of different computation models \cite{Karger00,Gabow91,Plotkin92,Karger93a,Karger96,GoranciHT16,Thorup07,Junger00,Chekuri97,Henzinger18,Chalermsook04,Karger99,Geissmann18,Kawarabayashi19,Henzinger17,Daga19,Geissmann17,Nanongkai14,Stoer97,Ghaffari13}. Its applications include network reliability~\cite{Ramanathan87,Karger95}, cluster analysis~\cite{Botafogo93}, and a critical subroutine in cutting-plane algorithms for the traveling salesman problem~\cite{Applegate07}.

A seminal result in weighted minimum cut algorithms is an algorithm by Karger~\cite{Karger00} which produces a minimum cut on an $m$-edge, $n$-vertex graph in $O(m \log^3 n)$ time with high probability\footnote{Probability $1-1/n^c$ for some constant $c$.}. This algorithm stood as the fastest minimum cut algorithm for the past two decades, until very recently, work published on arXiv shaved a log factor in Karger's approach~\cite{Mukhopadhyay19,Gawrychowski19}. The main component of Karger's algorithm is a subroutine that finds a minimum cut that $2$-respects (cuts two edges of) a given spanning tree $T$ of a graph $G$. In other words, the cut found is minimal amongst all cuts of $G$ that cut exactly two edges of $T$. Despite the number of pairs of spanning tree edges totaling $\Omega(n^2)$, Karger shows this can be accomplished in $O(m \log^2 n)$ time. Unfortunately, the procedure developed is particularly complex, a detail Karger admits when comparing the algorithm to a simpler $O(n^2 \log n)$ algorithm he develops to find \emph{all} minimum cuts~\cite{Karger00}. Indeed, perhaps for this reason, implementation of the asymptotically fastest minimum cut algorithm has been avoided in practical performance analyses~\cite{Chekuri97,Junger00}.


In this paper, we give a simple algorithm to find a minimum cut that $2$-respects a spanning tree $T$ of a graph $G$. Our procedure runs in $O(m \log^2 n)$ time, matching the performance of Karger's more-complicated subroutine. We achieve the simplification via a clever use of the heavy-light decomposition. Although our procedure requires the top tree data structure~\cite{AlstrupHLT2005} to achieve optimal performance, at the cost of an extra $O(\log n)$ factor, heavy-light decomposition can be used a second time so that only augmented binary search trees are required. We also give a self-contained version of Karger's algorithm~\cite{Karger00} with this new procedure and implement it, avoiding issues associated with previous implementations~\cite{Karger00,Chekuri97}.


Karger's algorithm~\cite{Karger00}, as well as the edge-sampling technique it is based on~\cite{Karger99}, has been extended and adapted to achieve results in a number of different settings~\cite{Ghaffari13,Daga19,Thorup07,Geissmann18,Geissmann17,Nanongkai14}. In particular, in the fully-dynamic setting, Thorup~\cite{Thorup07} uses the tree-packing technique developed by Karger~\cite{Karger00}, but maintains a larger set of trees so that the minimum cut $1$-respects at least one of them. In the parallel setting, Geissmann and Gianinazzi~\cite{Geissmann18} are able to parallelize both the dynamic tree data structure and the necessary computation required by Karger's algorithm~\cite{Karger00}. This work is based off prior work in the cache-oblivious model~\cite{Geissmann17}, also based on Karger's algorithm~\cite{Karger00}. In the distributed setting, Ghaffari and Kuhn~\cite{Ghaffari13} achieve a $(2+\epsilon)$-approximation to the minimum cut based on Karger's sampling technique~\cite{Karger99}. This is improved to a $(1+\epsilon)$-approximation with similar runtime by Nanongkai and Su~\cite{Nanongkai14}. Nanongkai and Su develop their algorithm from Thorup's fully-dynamic min-cut algorithm~\cite{Thorup07}, Karger's sampling technique~\cite{Karger99}, and Karger's dynamic program to find the minimum cut that $1$-respects a tree~\cite{Karger00}. Finally, Daga et al.~\cite{Daga19} achieve a sublinear time distributed algorithm to compute the exact minimum cut in an unweighted undirected graph. This algorithm builds off a more recent development in minimum cut algorithms~\cite{Kawarabayashi19}, combined again with the tree-packing technique introduced by Karger~\cite{Karger00}. Specifically, a tree packing is found in an efficient number of distributed rounds, then Karger's more-complicated algorithm to find a minimum $2$-respecting cut is applied in the distributed setting.

This vast amount of work based on Karger's original near-linear time algorithm suggests that simplifying it may yield additional techniques that can be applied both sequentially and in alternative settings. Indeed, the very recent improvements to Karger's algorithm~\cite{Mukhopadhyay19,Gawrychowski19} were published on arXiv two months after our paper was first made available online~\cite{Lovett19}, one of which~\cite{Gawrychowski19} cites our paper as what drew the authors to the problem. Indeed, their procedure for ``descendent edges'', given in Section 3.1, is similar to our procedure given in Section~\ref{two}. We have further found use of the approach given in this paper to achieve new results in dynamic higher connectivity algorithms~\cite{Molina19}.

This paper is organized as follows. In Section \ref{related}, we state the history of the minimum cut problem, in particular discussing other simple algorithms. In Section \ref{karger}, we give an overview of Karger's algorithm to pack spanning trees, leaving the details of the approach to Appendix~\ref{kargerdetails}. Our main contribution is given in Sections \ref{one} and \ref{two}. In Section \ref{one}, we show how to find minimum cuts that $1$-respect (cut one edge of) a tree using our new procedure. In Section \ref{two}, we extend the approach to find minimum cuts that $2$-respect (cut two edges of) a tree. We discuss our implementation in Section~\ref{implement} and give concluding remarks in Section~\ref{conclude}. 

\section{Related Work}
\label{related}

Before we begin, we give a brief history of the minimum cut problem. The minimum cut problem was originally perceived as a harder variant of the maximum $s$-$t$ flow problem and was solved by $n \choose 2$ flow computations. Gomory and Hu~\cite{GomoryH61} showed how to compute all pairwise max flows in $n-1$ flow computations, thus reducing the complexity of the minimum cut problem by a $\Theta(n)$ factor. Hao and Orlin~\cite{Hao94} further showed that the minimum cut in a directed graph can be reduced to a single flow computation.

Nagamochi and Ibaraki~\cite{Nagamochi92a,Nagamochi92b} developed a deterministic algorithm that is not based on computing maximum $s$-$t$ flows. They achieve $O(nm + n^2 \log n)$ time on a capacitated, undirected graph. This procedure was simplified by Stoer and Wagner~\cite{Stoer97}, achieving the same runtime. The Stoer-Wagner algorithm gives a simple procedure to find an \emph{arbitrary} minimum $s$-$t$ cut. Vertices $s$ and $t$ are then merged, and the procedure repeats. Although the $O(nm + n^2 \log n)$ time complexity requires an efficient priority queue such as a Fibonacci heap~\cite{Fredman87}, a binary heap can be used to achieve runtime $O(nm \log n)$.

Two algorithms based on \emph{edge contraction} have been devised. The first is an algorithm of Karger~\cite{Karger93a} and is incredibly simple. The algorithm randomly contracts edges until only two vertices remain. Repeated $O(n^2 \log n)$ times, the algorithm finds all minimum cuts on an undirected, weighted graph in $O(n^2 m \log n)$ time with high probability. This technique was improved by Karger and Stein~\cite{Karger96} by observing an edge of the minimum cut is more likely to be contracted later in the contraction procedure. Their improvement branches the contraction procedure after a certain threshold has been reached, spending more time to avoid contracting an edge of the minimum cut when fewer edges remain. The Karger-Stein algorithm achieves runtime $O(n^2 \log^3 n)$, finding the minimum cut with high probability.

In an unweighted graph, Gabow~\cite{Gabow91} showed how to compute the minimum cut in $O(cm\log(n^2/m))$ time, where $c$ is the capacity of the minimum cut. Karger~\cite{Karger99} improved Gabow's algorithm by applying random sampling, achieving runtime $\tilde{O}(m\sqrt{c})$ in expectation\footnote{The $\tilde{O}(f)$ notation hides $O(\log f)$ factors.}. The sampling technique developed by Karger~\cite{Karger99}, combined with the tree-packing technique devised by Gabow~\cite{Gabow91}, form the basis of Karger's near-linear time minimum cut algorithm~\cite{Karger00}. As previously mentioned, this technique finds the minimum cut in an undirected, weighted graph in $O(m \log^3 n)$ time with high probability.

A recent development uses low-conductance cuts to find the minimum cut in an undirected unweighted graph. This technique was introduced by Kawarabayashi and Thorup~\cite{Kawarabayashi19}, who achieve near-linear deterministic time (estimated to be $O(m \log^{12} n)$). This was improved by Henzinger, Rao, and Wang~\cite{Henzinger17}, who achieve deterministic runtime $O(m \log^2 n \;(\log \log n)^2)$. Although the algorithm of Henzinger et al. is more efficient than Karger's algorithm~\cite{Karger00} on unweighted graphs, the procedure, as well as the one it was based on~\cite{Kawarabayashi19}, are quite involved, thus making them largely impractical for implementation purposes.

Since an earlier version of this paper became available online~\cite{Lovett19}, several important improvements in minimum cut algorithms have been discovered. Ghaffari et al.~\cite{ghaffari20} devise a randomized unweighted minimum cut algorithm by using contraction based on sampling from each vertex, rather than standard uniform edge sampling. Their algorithm reduces unweighted minimum cuts to weighted minimum cuts on a graph with $O(n)$ edges, achieving $O(\min(m + n \log^3 n, m \log n))$ time complexity. Gawrychowski et al.~\cite{Gawrychowski19} improve Karger's procedure for finding the minimum cut that $2$-respects a tree to $O(m \log n)$ time. This improves the state-of-the-art for weighted minimum cuts to $O(m \log^2 n)$ time and, by Ghaffari et al.~\cite{ghaffari20}, improves the complexity of unweighted minimum cuts to $O(\min(m + n \log^2 n, m \log n))$ time. Mukhopadhyay and Nanongkai~\cite{Mukhopadhyay19} also study Karger's procedure for finding the minimum cut that $2$-respects a tree, arriving at an $O(m \frac{\log^2 n}{\log \log n} + n \log^6 n)$ time weighted minimum cut algorithm. Mukhopadhyay and Nanongkai further apply their new procedure to minimum cuts in the cut-query and streaming models.

\section{Overview of Karger's Spanning Tree Packing}
\label{karger}

We first formalize the definition mentioned earlier in this paper and originally given by Karger.

\begin{definition}[Karger~\cite{Karger00}]
\label{krespect}
Let $T$ be a spanning tree of $G$. We say that a cut in $G$ $k$-respects $T$ if it cuts at most $k$ edges of $T$. We also say that $T$ $k$-constrains the cut in $G$.
\end{definition}

We also define weighted tree packings.

\begin{definition}[Karger~\cite{Karger00}]
A weighted tree packing is a set of spanning trees, each with an assigned non-negative weight, such that the total weight of trees containing a given edge of $G$ is no greater than the weight of that edge. The weight of the packing is the total weight of the trees in it.
\end{definition}

The first stage of Karger's algorithm is to sample edges independently and uniformly at random from graph $G$ to form a graph $H$, and then pack spanning trees in $H$. If we sample a tree $T$ from a packing with probability proportional to its weight, a minimum cut in $G$ will cut at most two edges of $T$ with constant probability. Thus, if we sample $O(\log n)$ trees from the weighted packing, a minimum cut in $G$ $2$-respects at least one of the sampled trees with high probability. The remainder of the algorithm is a procedure that, given a spanning tree $T$ of a graph $G$, finds a minimal cut of $G$ that $2$-respects $T$. This procedure is applied to all $O(\log n)$ sampled spanning trees.

We leave the intuition behind Karger's approach and the relevant mathematics to Appendix~\ref{kargerdetails}. We will use Algorithm~\ref{packpart} to pack spanning trees, credited to Thorup and Karger~\cite{Thorup00}, Plotkin-Shmoys-Tardos~\cite{Plotkin92}, and Young~\cite{Young95}. The procedure appears in Gawrychowski et al.~\cite{Gawrychowski19}.

\begin{algorithm}[H]
\caption{Obtain a Packing of Weight at least $.4c$ from a Graph $G$}
\label{packpart}
\begin{flushleft}
Let $G$ be a graph with $m$ edges and $n$ vertices.
\begin{enumerate}
  \item Initialize $\ell(e) \leftarrow 0$ for all edges $e$ of $G$. Initialize multiset $P \leftarrow \emptyset$. Initialize $W \leftarrow 0$.
  \item Repeat the following:
  \begin{enumerate}
    \item Find a minimum spanning tree $T$ with respect to $\ell(\cdot)$.
    \item Set $\ell(e) \leftarrow \ell(e) + 1/(75 \ln m)$ for all $e \in T$. If $\ell(e) > 1$, return $W, P$.
    \item Set $W \leftarrow W + 1/(75 \ln m)$.
    \item Add $T$ to $P$.
  \end{enumerate}
\end{enumerate}
\end{flushleft}
\end{algorithm}

\begin{lemma}[\cite{Plotkin92,Thorup00,Young95}]
\label{packpartlemma}
Given an undirected unweighted graph $G$ with $m$ edges, $n$ vertices, and minimum cut $c$, Algorithm~\ref{packpart} returns a weighted packing of weight at least $.4c$ in $O(mc \log n)$ time.
\end{lemma}

Algorithm~\ref{packpart} and Lemma~\ref{packpartlemma} are given in Appendix~\ref{kargerdetails} with general epsilon and proven. To achieve $O(mc \log n)$ time in Algorithm~\ref{packpart}, we may use a linear time minimum spanning tree routine~\cite{Karger95c} or the following implementation trick given by Gawrychowski et al.~\cite{Gawrychowski19}. In the use of Algorithm~\ref{packpart} in Algorithm~\ref{packing}, the graph in Algorithm~\ref{packpart} has edges which may be duplicated $O(\log n)$ times, while the number of distinct edges can be bounded as a factor $\Theta(\log n)$ fewer. It suffices to invoke the minimum spanning tree algorithm of Algorithm~\ref{packpart} with only the minimum of each set of parallel edges. We can easily maintain the minimum of each set of parallel edges in $O(\log n)$ time per edge per iteration, which suffices to shave a log factor in the runtime of Algorithm~\ref{packpart}. Note that if we chose to avoid these optimizations and/or avoid the use of top trees in Section~\ref{two}, the final runtime becomes $O(m \log^4 n)$.

We use Algorithm~\ref{packpart} in Algorithm~\ref{packing} to obtain $\Theta(\log n)$ trees for the $2$-respect algorithm given in Sections~\ref{one} and~\ref{two}.

\begin{algorithm}[H]
\caption{Obtain $\Theta(\log n)$ Spanning Trees for the $2$-respect Algorithm}
\label{packing}
\begin{flushleft}
Let $d$ denote the exponent in the probability of success $1-1/n^d$. Let $b = 3 \cdot 6^2(d+2) \ln n$.
\begin{enumerate}
  \item Form graph $G'$ from $G$ by first normalizing the edge weights of $G$ so the smallest non-zero edge weight has weight $1$, then multiplying each edge weight by $100$ and rounding to the nearest integer. Let $U$ be an upper bound for the size of the minimum cut of $G'$.
  \item Initialize $c' \leftarrow U$. Repeat the following:
  \begin{enumerate}
    \item Construct $H$ in the following way: for each edge $e$ of $G'$, let $e$ have weight in $H$ drawn from the binomial distribution with probability $p = \min(b/c', 1)$ and number of trials the weight of $e$ in $G'$. Cap the weight of any edge in $H$ to at most $\lceil 7/6 \cdot 12b \rceil$.
    \item Run Algorithm~\ref{packpart} on $H$, considering an edge of weight $w$ as $w$ parallel edges. There are three cases:
    \begin{enumerate}
    	\item If $p = 1$, set $P$ to the packing returned and skip to step 3.
    	\item If the returned packing is of weight $24b/70$ or greater, set $c' \leftarrow c'/6$ and repeat steps 2a and 2b, setting $P$ to the packing returned and then proceeding to step 3.
    	\item Otherwise, repeat steps 2a and 2b with $c' \leftarrow c'/2$.
    \end{enumerate}
  \end{enumerate}
    \item Return $\lceil36.53d\ln n\rceil$ trees sampled uniformly at random proportional to their weights from $P$.
\end{enumerate}
\end{flushleft}
\end{algorithm}

\begin{lemma}
\label{packinglemma}
Algorithm \ref{packing} returns a collection of $\Theta(\log n)$ spanning trees of $G$ in time $O(m \log^3 n)$ such that the minimum cut of $G$ $2$-respects at least one tree in the collection with high probability.
\end{lemma}

Algorithm~\ref{packing} and Lemma~\ref{packinglemma} are given in Appendix~\ref{kargerdetails} with general epsilon and proven.

\section{Minimum Cuts that $1$-Respect a Tree}
\label{one}

We now give our algorithm for finding a minimum cut that $1$-respects a spanning tree $T$ of a graph $G$. We present it here only to build intuition for the idea used to find $2$-respecting cuts in the following section, which also finds $1$-respecting cuts. 

We use the following lemma, a consequence of Sleator and Tarjan's heavy-light decomposition~\cite{SleatorTarjan83}.
\begin{lemma}[Sleator and Tarjan~\cite{SleatorTarjan83}]
\label{order}
Given a tree $T$, there is an ordering of the edges of $T$ such that the edges of the path between any two vertices in $T$ consist of the union of up to $2\log n$ contiguous subsequences of the order. The order can be found in $O(n)$ time.
\end{lemma}
\begin{proof}
We use heavy-light decomposition, credited to Sleator and Tarjan~\cite{SleatorTarjan83}. Note that the algorithm assumes $T$ is rooted. We can root $T$ arbitrarily. We then take the heavy paths given from the usual construction and concatenate them in any order.
\end{proof}

Our algorithm begins by labeling the edges of $T$ in heavy-light decomposition order $e_1, \ldots, e_{n-1}$ as given by Lemma~\ref{order}. Consider the cut of $G$ induced by the vertex partition resulting from cutting a single edge of $T$. We iterate index $i$ through heavy-light decomposition order and keep up-to-date the total weight of all edges of $G$ that cross the cut induced by $e_i$. The minimum weight found is then returned.

Call the edges of $G$ in $T$ \emph{tree edges} and edges of $G$ not in $T$ \emph{non-tree edges}. Critical to our approach is the following proposition.
\begin{proposition}
\label{nontree}
For any cut of $G$ that $2$-respects $T$, the non-tree edge $uv$ crosses the cut if and only if exactly one tree edge from the $uv$-path in $T$ crosses the cut.
\end{proposition}
\begin{proof}
Recall that for any edge of $T$ crossing the cut, the components of each of its endpoints must fall on opposite sides of the cut. Therefore if the number of tree edges in the cut on the $uv$-path in $T$ is odd, the non-tree edge $uv$ crosses the cut. Since we are only considering cuts that cut at most $2$ edges of $T$, the proposition follows.
\end{proof}

We now give our algorithm explicitly.
\begin{algorithm}[H]
\caption{Minimum Cuts that $1$-Respect $T$}
\label{1respect}
\begin{flushleft}
\begin{enumerate}
  \item Arrange the edges of $T$ in the order of Lemma \ref{order}; label them $e_1, \ldots, e_{n-1}$.
  \item For each non-tree edge $uv$, mark every $i$ such that $e_i$ is on the $uv$-path in $T$ and $e_{i+1}$ is not on the $uv$-path in $T$, or vice versa. Indicate whether edge $e_1$ is on the $uv$-path in $T$.
  \item Iterate index $i$ from $1$ to $n-1$, in each iteration keeping track of the total weight of all non-tree edges $uv$ such that $e_i$ lies on the $uv$-path in $T$, added together with the weight of edge $e_i$.
  \item Return the minimum total weight found in step 3.
\end{enumerate}
\end{flushleft}
\end{algorithm}

\begin{lemma}
\label{1respectlemma}
Algorithm~\ref{1respect} finds the value of the minimum cut that $1$-respects a spanning tree $T$ of a graph $G$ in $O(m \log n)$ time.
\end{lemma}
\begin{proof}
Via Proposition~\ref{nontree}, in a $1$-respecting cut including only $e_i$ from $T$, a non-tree edge $uv$ is cut if and only if the edge $e_i$ lies on the $uv$-path in $T$. Algorithm~\ref{1respect} keeps track of all such non-tree edges for each possible $e_i$ that is cut, therefore it finds the minimum cut of $G$ that cuts a single edge of $T$.

The time complexity can be determined as follows. Finding the heavy-light decomposition for step 1 takes $O(n)$ time. In doing so, we can label each edge and each heavy path so that every edge knows its index in the order as well as the heavy path to which it belongs. Each heavy path can store its starting and ending index in the order. With this information, step 2 can be completed by walking up from $u$ and $v$ in $T$ towards the root of $T$. We spend $O(1)$ work per heavy path from root to vertex, which is bounded by $O(\log n)$ via the heavy-light decomposition. In total this step takes $O(m \log n)$ time.

In step 3, we spend $O(n)$ total work plus $O(1)$ work for each transition of the current edge $e_i$ on or off the $uv$ path for all non-tree edges $uv$. Each non-tree edge transitions on or off $O(\log n)$ times as guaranteed by Lemma~\ref{order}, therefore the time complexity of this step is $O(m \log n)$. Overall, Algorithm \ref{1respect} takes $O(m \log n)$ time.
\end{proof}

Note that if we wish to find the edges in the minimum cut, we can keep track of the minimum-achieving index $i$ so we know the vertex separation of the minimum cut. With the vertex separation, it is easy to find in $O(m \log n)$ time which non-tree edges cross the cut.

Further note that we need not know the identity of the non-tree edge $uv$ as $e_i$ falls on or off the $uv$-path. Thus the space required for step 2 need only be $O(m)$, since at each transition point we can just keep track of the total weight added or subtracted from the minimum cut.

\section{Minimum Cuts that $2$-Respect a Tree}
\label{two}

We now discuss an extension of Algorithm~\ref{1respect} to find a minimum cut that $2$-respects a tree. We still iterate $i$ through heavy-light decomposition order, but in addition to cutting $e_i$, we find the best $j$ so that the cut resulting from cutting $e_i$ and $e_j$ is minimal. To find the best $j$ efficiently we use a clever data structure.

\begin{lemma}[Alstrup et al.~\cite{AlstrupHLT2005}]
\label{top}
There is a data structure that supports the following operations on a weighted tree $T$ in $O(\log n)$ time:
\begin{itemize}
  \item \texttt{PathAdd(u, v, x)} := Add weight $x$ to all edges on the unique $uv$-path in $T$.
  \item \texttt{NonPathAdd(u, v, x)} := Add weight $x$ to all edges not on the unique $uv$-path in $T$.
  \item \texttt{QueryMinimum()} := Query for the minimum weight edge in $T$.
\end{itemize}
\end{lemma}
\begin{proof}
Operations \texttt{PathAdd()} and \texttt{QueryMinimum()} are just Theorems 3 and 4 of \cite{AlstrupHLT2005}. Operation \texttt{NonPathAdd(u, v, x)} can be achieved by keeping a counter of global weight added to (subtracted from) $T$ and executing \texttt{PathAdd(u, v, -x)} to undo this action on the $uv$-path. See also~\cite{Tarjan05}.
\end{proof}

Note that the weight $x$ can be positive or negative.

If we seek to avoid implementing any sophisticated data structures, we can instead use heavy-light decomposition again and support the above two operations in $O(\log^2 n)$ time. To see how, by Lemma~\ref{order} each path of $T$ represents at most $O(\log n)$ contiguous segments of the total order of edges. Range add and a global minimum query can be supported in $O(\log n)$ time via an augmented binary search tree. 
Thus the total time complexity per operation is $O(\log^2 n)$.

We use the range operations as follows. As we iterate index $i$ through the order of Lemma~\ref{order}, we keep up to date the cost of the cut resulting from cutting any other edge $e_j$ via the data structure of Lemma~\ref{top}. Instead of querying each other edge $e_j$ directly, however, we just use a global minimum query to find the best choice of $j$. The procedure is given in Algorithm~\ref{2respect}. The first two steps are the same as Algorithm~\ref{1respect}.

\begin{algorithm}
\caption{Minimum Cuts that $2$-Respect $T$}
\label{2respect}
\begin{flushleft}
\begin{enumerate}
  \item Arrange the edges of $T$ in the order of Lemma \ref{order}; label them $e_1, \ldots, e_{n-1}$.
  \item For each non-tree edge $uv$, mark every $i$ such that $e_i$ is on the $uv$-path in $T$ and $e_{i+1}$ is not on the $uv$-path in $T$, or vice versa. Indicate whether edge $e_1$ is on the $uv$-path in $T$.
  \item Initialize the data structure of Lemma~\ref{top} on $T$ so that the weight of edge $e_j$ is equal to its weight in $T$.
  \item Iterate index $i$ from $1$ to $n-1$. Via the computation done in step 2, maintain the following invariants in the data structure of Lemma~\ref{top} as $i$ is iterated.
  \begin{enumerate}
    \item When edge $e_i$ is on the $uv$-path in $T$, add the weight of non-tree edge $uv$ to all edges off the $uv$-path in $T$.
    \item When edge $e_i$ is off the $uv$-path in $T$, add the weight of non-tree edge $uv$ to all edges on the $uv$-path in $T$.
  \end{enumerate}
Each time $i$ is incremented, after updating weights in Lemma~\ref{top} as per 4a and 4b, add $\infty$ to edge $e_i$, execute \texttt{QueryMinimum()}, then subtract $\infty$ from edge $e_i$. The value of the minimum cut found in each iteration is the result of \texttt{QueryMinimum()} plus the weight of $e_i$.
\item Return the minimum of the smallest cut found in step 4 with the result of \texttt{QueryMinimum()} when we consider edge $e_i$ to be off the path of all non-tree edges $uv$ in the data structure of Lemma~\ref{top}.
\end{enumerate}
\end{flushleft}
\end{algorithm}

\begin{lemma}
\label{2respectlemma}
Algorithm~\ref{2respect} finds the value of the minimum cut that $2$-respects a spanning tree $T$ of a graph $G$ in $O(m \log^2 n)$ time.
\end{lemma}
\begin{proof}
By Proposition~\ref{nontree}, in a $2$-respecting cut including $e_i$ and $e_j$ of $T$, a non-tree edge $uv$ is cut if and only if exactly one of $e_i$ or $e_j$ lies on the $uv$-path in $T$. Observe that the invariants enforced in step 4 guarantee that in each iteration the total weight of edges from the cut resulting from cutting any other edge $e_j$ along with $e_i$ is kept up-to-date in the data structure of Lemma~\ref{top}. Since the minimum such $j$ is found for every $i$, it follows that step 4 finds the weight of the minimum cut of $G$ that cuts exactly two edges of $T$. In step 5, we return the minimum of this weight with a single call to \texttt{QueryMinimum()} where we assume edge $e_i$ to be off the path of all non-tree edges $uv$. Observe that this computes the minimum cut of $G$ that cuts exactly one edge of $T$. Thus, the minimum cut of $G$ that $2$-respects $T$ is returned in step 5.

The time complexity follows similarly to Algorithm~\ref{1respect}. Steps 1 and 2 take $O(m \log n)$ total time. However, step 4 requires $O(\log n)$ time for non-tree edge $uv$ whenever edge $e_i$ falls on or off the $uv$-path in $T$, since the data structure of Lemma~\ref{top} takes $O(\log n)$ time per operation. For a given non-tree edge $uv$, edge $e_i$ falls on or off the $uv$-path in $T$ a total of $O(\log n)$ times by Lemma~\ref{order}; thus step 4 takes $O(m \log^2 n)$ time. The final \texttt{QueryMinimum()} call in step 5 takes $O(m \log n)$ time. The total time taken is $O(m \log^2 n)$.
\end{proof}

We make a few further remarks about Algorithm~\ref{2respect}. To determine the edges of the minimum cut, the data structure of Lemma~\ref{top} can be augmented to return the index $j$ of the edge that achieves the minimum given in operation \texttt{QueryMinimum()}. With $e_i$ and $e_j$, we can determine the vertex partition in $G$ of the minimum cut and, as stated in Section~\ref{one}, and from this we can find which non-tree edges cross the minimum cut easily in $O(m \log n)$ time.

The space complexity of Algorithm~\ref{1respect} was easily linear. In Algorithm~\ref{2respect}, we must know the identity of each non-tree edge $uv$ in every transition point where edge $e_i$ falls on or off the $uv$-path. Naively this costs $O(m \log n)$ space. This can be improved to $O(m)$ space by performing step 2 incrementally while executing step 4. That is, we only need to know the next transition point where the non-tree edge $uv$ falls on or off the $uv$-path, and from the current transition point this can be determined in constant time via the heavy-light decomposition.

Recall that while Algorithm~\ref{1respect} helped demonstrate the approach of Algorithm~\ref{2respect}, we need only implement Algorithm~\ref{2respect}, since Algorithm~\ref{2respect} finds the minimum cut of $G$ that cuts either $1$ or $2$ edges of $T$.

From this we get our final theorem, equivalent to the result of Karger~\cite{Karger00}.

\begin{theorem}
The minimum cut in a weighted undirected graph can be found in $O(m \log^3 n)$ time with high probability.
\end{theorem}
\begin{proof}
We first find $\Theta(\log n)$ spanning trees by Algorithm~\ref{packing}. We then find the minimum cuts that $2$-respect each of these trees by Algorithm~\ref{2respect}. By Lemmas \ref{packinglemma} and \ref{2respectlemma}, this finds the minimum cut with high probability in $O(m \log^3 n)$ time.
\end{proof}

\section{Implementation}
\label{implement}

We have implemented an $O(m \log^4 n)$ version of our algorithm in C++\footnote{Our implementation is available at: \url{https://github.com/nalinbhardwaj/min-cut-paper}.}. Algorithm~\ref{packpart} together with an $O(m \log n)$ minimum spanning tree routine take about 100 lines of code, Algorithm~\ref{packing} takes about 200 lines, Algorithm~\ref{2respect} takes about 200 lines, and using an augmented binary search tree as the data structure for Lemma~\ref{top} takes about 200 lines. To the best of our knowledge, our implementation is the first to achieve near-linear time complexity. We have tested it against an $O(n^3)$ implementation of the Stoer-Wagner algorithm~\cite{Stoer97} and an $O(n^3 \log n)$ implementation of Karger's randomized contraction algorithm~\cite{Karger93a}. Under favorable inputs, the runtime compares as in Figure~\ref{graph}.

\begin{figure}
	\centering
	\begin{tikzpicture}
	\begin{axis}[legend pos=outer north east, title={\textbf{Performance Comparison}},
	xlabel={$n$}, ylabel={Runtime in seconds}]
	\addplot+[smooth,mark=*] plot coordinates
	{ 	(1000, 36.741) 
		(1100, 42.764)
		(1200, 47.679)
		(1300, 53.378)
		(1400, 55.860)
		(1500, 62.457)
		(1600, 64.776)
		(1700, 60.313)
		(1800, 75.481)
		(1900, 78.935)
		(2500, 131.265)
	};
	\addlegendentry{Our Algorithm}
	\addplot+[smooth,mark=x] plot coordinates
	{ 	(1000, 20.397) 
		(1100, 29.253)
		(1200, 35.580)
		(1300, 45.768)
		(1400, 56.803)
		(1500, 70.090)
		(1600, 82.602)
		(1700, 98.456)
		(1800, 116.289)
		(1900, 140.121)
		(2500, 318.721) };
	\addlegendentry{Stoer-Wagner~\cite{Stoer97}}
	\addplot+[smooth,mark=+] plot coordinates
	{ 	(1000, 19.482) 
		(1100, 30.780)
		(1200, 42.528)
		(1300, 58.399)
		(1400, 73.294)
		(1500, 89.544)
		(1600, 106.710)
		(1700, 124.241)
		(1800, 133.266)
		(1900, 160.230)
		(2500, 322.708) };
	\addlegendentry{Karger~\cite{Karger93a}}
	\end{axis}
	\end{tikzpicture}
	\caption{Performance comparison of an $O(m \log^4 n)$  implementation of our algorithm with an $O(n^3)$ Stoer-Wagner~\cite{Stoer97} and $O(n^3 \log n)$ Karger~\cite{Karger93a}.}
	\label{graph}
\end{figure}
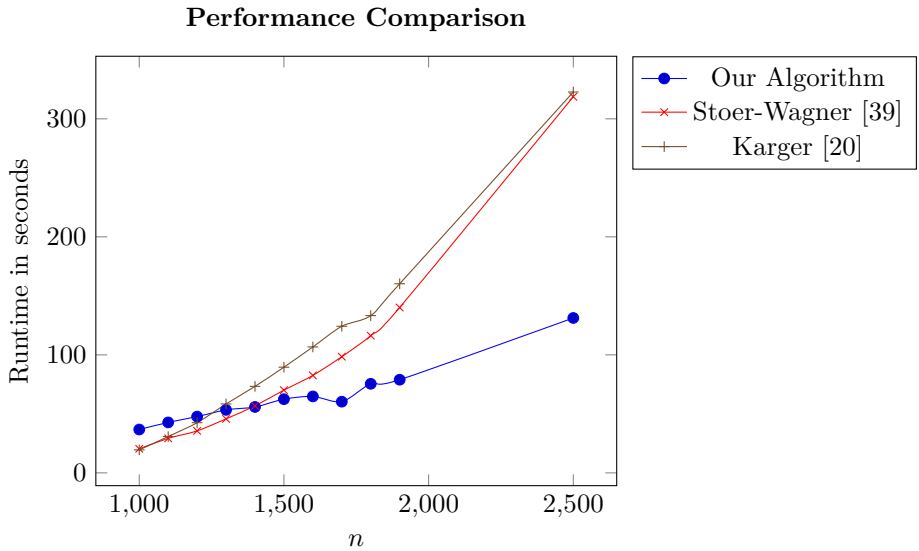

Figure~\ref{graph} demonstrates the near-linear growth in the running time of our algorithm. Unfortunately, it does not appear our implementation is competitive compared to existing implementations~\cite{Chekuri97}. The bottleneck is in obtaining the $O(\log n)$ spanning trees for Algorithm~\ref{2respect}, even when Algorithm~\ref{packing} runs in $O(m \log^3 n)$ time and Algorithm~\ref{2respect} runs in $O(m \log^4 n)$ time. The issue is the large constant factors due to the quadratic dependencies on epsilons, seen in Algorithms~\ref{packpart2} and \ref{packing2}. We have calculated that the number of calls to the minimum spanning tree routine in our implementation can be as much as $8100 \ln n \ln m$, and that changing the choices of epsilons for Algorithms~\ref{packpart} and \ref{packing} does not yield significant improvement.

If we replace Algorithm~\ref{packpart} with the more-complicated Gabow's algorithm~\cite{Gabow91}, we can likely improve our implementation's runtime. Further, a factor of about two can be saved by finding $c'$ via an approximation algorithm~\cite{Karger95b}. However, a large constant factor will remain due to the sampling procedure in Lemma~\ref{sample}, discussed in Appendix~\ref{kargerdetails}. All known algorithms to compute weighted tree packings have dependence on $c$, the value of the minimum cut, and Lemma~\ref{sample} reduces the value of the minimum cut to at least $3(d+2)(\ln n)/\epsilon^2$, which in our algorithms manifests as a factor of $108(d+2) \ln n$. It appears that for Karger's approach to be made practical, this large constant factor will likely need to be improved or heuristic approaches would need to be considered~\cite{Chekuri97}.

\section{Conclusion}
\label{conclude}

In this paper, we have discussed a simplification to Karger's original near-linear time minimum cut algorithm~\cite{Karger00}. In contrast to Karger's original algorithm~\cite{Karger00}, finding spanning trees that have a constant probability of $2$-respecting the minimum cut is now the more-complicated part of the algorithm and finding minimum cuts that $2$-respect a tree is relatively simpler. In actuality, both were complicated in Karger's original algorithm, however the work to find the tree packing was largely abstracted to previous publications. The same can be said for many statements of Karger's near-linear time algorithm~\cite{Gawrychowski19,Mukhopadhyay19}. Our version, on the other hand, is self-contained: the only procedures outside of Algorithms~\ref{packpart}, \ref{packing}, and \ref{2respect} required to implement the full algorithm are a minimum spanning tree subroutine and (optionally) a top tree data structure.

The main contribution of our algorithm is a new, simple procedure to find a minimum cut that $2$-respects a tree $T$ in $O(m \log^2 n)$ time. Karger advertises that the complexity of his near-linear time algorithm is $O(m \log^3 n)$ and thus his routine to find a minimum cut that $2$-respects a tree also takes $O(m \log^2 n)$ time. However, he gives two small improvements to the algorithm to reduce the overall runtime to $O(m \log^2 n \log(n^2/m)/\log \log n + n \log^6 n)$. The first uses the fact that finding a $1$-respecting cut can be done in linear time, and the other is an improvement which reduces an $O(\log n)$ factor to an $O(\log (n^2/m))$ factor in the $2$-respect routine. For our algorithm, the first improvement can be applied by substituting our $1$-respect algorithm with his. The second improvement can not be applied. Thus, when $m = \Theta(n^2)$, his algorithm is faster by an $O(\log n)$ factor. However, for this case, Karger gives a different, simpler algorithm~\cite{Karger00} which finds the global minimum cut in $O(n^2 \log n)$ time anyway.

There are three algorithms that are referred to as simple min-cut algorithms: the Stoer-Wagner algorithm~\cite{Stoer97} which runs in $O(nm \log n)$ time or $O(nm + n^2 \log n)$ time with a Fibonacci heap~\cite{Fredman87}, Karger's randomized contraction algorithm~\cite{Karger93a} which runs in $O(n^2m \log n)$ time, and the improvement to Karger's algorithm by Karger and Stein~\cite{Karger96} which runs in $O(n^2 \log^3 n)$ time. In comparison to these, our approach is the least simple. However, our $O(m \log^3 n)$ runtime is significantly better. While the large constant factors in our approach make this only relevant at large values of $n$, we hope the procedure developed in this paper can be used in conjunction with an optimized version of Karger's sampling technique to produce an asymptotically fast, practical minimum cut algorithm.

\bibliography{KargerSimple}

\appendix

\section{Karger's Algorithm for Packing Spanning Trees}
\label{kargerdetails}

In this section we give the intuition and mathematics behind the spanning tree packing of Karger's algorithm.

\subsection{Tree Packing}

The basic idea of Karger's near-linear time algorithm~\cite{Karger00} is to exploit the following combinatorial result. Recall that a tree packing of an undirected unweighted graph $G$ is a set of spanning trees such that each edge of $G$ is contained in at most one spanning tree. The weight of a tree packing is the number of trees in it.

\begin{theorem}[Nash-Williams~\cite{NashWilliams61}]
\label{nashw}
Any undirected unweighted multigraph with minimum cut $c$ contains a tree packing of weight at least $c/2$.
\end{theorem}

Now consider a minimum cut and a tree packing given by Theorem~\ref{nashw}. Each edge of the minimum cut can only be present in at most one spanning tree. As there are $c$ edges of the minimum cut, this implies that the average spanning tree contains at most $c/(c/2) = 2$ edges of the minimum cut. In other words, a spanning tree chosen at random from a packing of Theorem~\ref{nashw} will $2$-constrain the minimum cut with probability at least $1/2$.

Suppose we are given a spanning tree $T$ of $G$ with each edge of $T$ marked if it crosses the minimum cut. The endpoints of any marked edge must fall on opposite sides of the cut. Conversely, the endpoints of any unmarked edge must be on the same side of the cut. It follows that if we know the edges of $T$ that cross the minimum cut, we can determine the vertex partition of the minimum cut and its total weight in $G$.

This gives the intuition behind Karger's algorithm~\cite{Karger00}. We sample spanning trees from a tree packing of $G$ and for each tree $T$, we find the minimum cut that $2$-respects $T$. Unfortunately, several obstacles need be overcome before this can be made into an efficient algorithm. For one, all currently known approaches of determining a tree packing of Theorem~\ref{nashw} have runtime $\Omega(cm)$, which for large values of $c$ is far more than the runtime we seek. Further, Theorem~\ref{nashw} must be generalized to weighted graphs.

We first address the latter concern. Recall the definition of weighted tree packings given in Section~\ref{karger}.

\begin{lemma}[Karger~\cite{Karger00}]
\label{weighted-nashw}
Any undirected weighted graph with minimum cut $c$ contains a weighted tree packing of weight at least $c/2$.
\end{lemma}
\begin{proof}
For contradiction, suppose some graph $G$ with minimum cut $c$ and $\epsilon > 0$ exist such that $G$ does not contain a weighted packing of weight $(1-\epsilon)c/2$ or greater.

Take $G$ and approximate each edge $e_i$ of weight $w_i$ by a rational number $a_i/b_i$ such that $a_i/b_i < w_i$ and $w_i - a_i/b_i < \epsilon$. Multiply all edges by $d = \prod_i b_i$ and call the resulting graph $G'$. Then by Theorem~\ref{nashw}, when viewed as an unweighted multigraph, $G'$ has a tree packing of weight at least $(1-\epsilon)dc/2$. If we weight each tree of the packing by $1/d$, the packing becomes a weighted packing of $G$ of weight at least $(1-\epsilon)c/2$, a contradiction.
\end{proof}

Note that for both Lemma~\ref{nashw} and Lemma~\ref{weighted-nashw}, an upper bound of weight $c$ also exists, because every spanning tree in the packing must cross the minimum cut at least once.

To effectively use Lemma~\ref{weighted-nashw}, we formally state the relationship between weighted packings and trees that $2$-constrain small cuts.

\begin{lemma}[Karger~\cite{Karger00}]
\label{cuts}
Consider a weighted graph $G$ and a weighted tree packing of weight $\beta c$, where $c$ is the weight of the minimum cut in $G$. Then given a cut of weight $\alpha c$, a fraction at least $\frac{1}{2}(3 - \alpha/\beta)$ of the trees (by weight) $2$-constrain the cut.
\end{lemma}
\begin{proof}
Note that every spanning tree must cross every cut.
Let $x$ denote the total weight of trees with at least three edges crossing the cut and $y$ the total weight of trees with one or two edges crossing the cut.
Then $x+y=\beta c$ and $3x+y\le\alpha c$. Rearranging, we get $y\ge\frac12(3\beta c-\alpha c)$.
\end{proof}

\subsection{Random Sampling}

In order to avoid the $\Omega(cm)$ complexity of finding a packing of weight $c/2$, we first apply random sampling to $G$. Specifically, we use the following from Karger's earlier work.

\begin{lemma}[Karger \cite{Karger99}]
\label{sample}
Let $p = 3(d+2)(\ln n)/(\epsilon^2 \gamma c) \leq 1$, where $c$ is the weight of the minimum cut of an unweighted multigraph $G$ and $\gamma \leq 1, \gamma = \Theta(1)$. Then if we sample each edge of $G$ independently with probability $p$, the resulting graph $H$ has the following properties with probability $1-1/n^d$.
\begin{enumerate}
  \item The minimum cut in $H$ is of size within a $(1+\epsilon)$ factor of $cp = 3(d+2)(\ln n)/(\gamma \epsilon^2)$, which is $O(\epsilon^{-2} \log n)$.
  \item A cut in $G$ takes value within a factor $(1+\epsilon)$ of its expected value in $H$. In particular, the minimum cut in $G$ corresponds (under the same vertex partition) to a $(1+\epsilon)$-times minimum cut of $H$.
\end{enumerate}
\end{lemma}

By picking $\epsilon$ to be a constant such as $1/6$, Lemma~\ref{sample} will allow us to reduce the size of the minimum cut in $H$ to $O(\log n)$. We can then run existing algorithms~\cite{Plotkin92,Gabow91} to pack trees in $H$ in $\tilde{O}(m)$ time. Further, since the minimum cut of $G$ corresponds to a $(1+\epsilon)$-times minimum cut of $H$, we can still apply Lemma~\ref{cuts} on the sampled graph $H$ so that a tree randomly sampled from the packing has a constant probability of $2$-constraining the minimum cut in $G$.

There are still several issues to resolve. Lemma~\ref{sample} applies to unweighted multigraphs $G$, but our graph $G$ can have non-negative real weights. The other issue is that the value $\gamma$ needs to be known ahead of time in order to apply the lemma. We first address the latter issue.

Lemma~\ref{sample} requires knowing a constant-factor underestimate $c' = \gamma c$ for the minimum cut $c$. In particular, without $\gamma \leq 1$, property 2 of Lemma~\ref{sample} is not guaranteed with high probability, and if $\gamma = o(1)$, the minimum cut of $H$ will be of size $\omega(\epsilon^{-2} \log n)$ with high probability. We may run a linear-time $3$-approximation algorithm~\cite{Matula93}, with modifications to work on weighted graphs~\cite{Karger95b}, to find this approximation. This is simple to state, but more difficult to implement.

A different approach is to start with a known upper bound $U$ for $c'$. Karger states that we can then halve this upper bound until ``our algorithms succeed''~\cite{Karger99}. This approach is taken by the implementation of Chekuri et al.~\cite{Chekuri97}. Unfortunately, it is not rigorous as stated. Lemma~\ref{sample} indicates that with a constant-factor underestimate $c' = \gamma c$ for $c$, our algorithm can proceed. However, it does not give a process for rejecting a guess $c'$ that is not a constant-factor underestimate for $c$. We could try all powers of $2$ for $c'$ within a known lower and upper bound of the value of the minimum cut, and run our algorithms for all possibilities. This is rigorous, but introduces an extra $O(\log n)$ factor in our runtimes, assuming the range of $c'$ we try is polynomial in $n$. We instead show the following.

\begin{lemma}
\label{overestimate}
Let $p = 3(d+2)(\ln n)/(\epsilon^2 \gamma c) \leq 1$ as in Lemma~\ref{sample}, but with $\gamma \geq 6$ and $\epsilon \leq 1/3$. Then if we sample each edge of the unweighted multigraph $G$ uniformly at random with probability $p$, the resulting graph $H$ has minimum cut of size less than $(d+2)(\ln n)/\epsilon^2$ with probability at least $1-1/n^{d+2}$.
\end{lemma}

\begin{proof}
Consider the size of a minimum cut of $G$ as a cut in $H$. Let $X$ be a random variable denoting this size. Then $\Ex{X} = cp$. By a Chernoff bound, $\Pr{X \geq (1+\delta) cp} \leq e^{-\frac{1}{3}(cp\delta)}$ for $\delta \geq 1$. Let $(1+\delta) = \frac{\gamma}{3}$. Then
\begin{align*}
  \Pr{X \geq (d+2)(\ln n)/\epsilon^2} &\leq e^{-\frac{1}{3}(cp(\frac{\gamma}{3}-1))}\\
  &= e^{-(d+2)(\ln n)\gamma^{-1}\epsilon^{-2}(\frac{\gamma}{3}-1))}\\
  &= n^{-\frac{1}{3}(d+2)\epsilon^{-2} + (d+2)\gamma^{-1}\epsilon^{-2}}\\
  &\leq n^{-\frac{1}{6}(d+2)\epsilon^{-2}}\\
  &< n^{-(d+2)}.
\end{align*}
Therefore, the minimum cut in $H$ has size less than $(d+2)(\ln n)/\epsilon^{-2}$ with probability at least $1-1/n^{d+2}$.
\end{proof}

Lemma~\ref{overestimate} states that if our estimate $c' = \gamma c$ satisfies $\gamma \geq 6$, the minimum cut will be at least a factor $3$ smaller than $3(d+2)(\ln n)/\epsilon^2$ with high probability. Recall that with $\gamma = 1$ and therefore $c' = c$, we expect the minimum cut in $H$ to be within a factor $(1+\epsilon)$ from $3(d+2)(\ln n)/\epsilon^2$ with high probability. Lemma~\ref{overestimate} gives us the necessary tool to reject $c'$ that are not a constant factor underestimate of $c$. We try a value for $c'$, and if the size of the minimum cut in $H$ is greater than $(1+\epsilon)^{-1}3(d+2)(\ln n)/\epsilon^2$, we know $c' < 6c$. Therefore we can decrease $c'$ by a factor of $6$ and rerun the tree packing algorithm. The resulting graph $H$ must satisfy the conditions of Lemma~\ref{sample}, therefore the algorithm may proceed. Since our tree packing algorithms determine the minimum cut up to constant factors, this approach avoids the need of a different (or recursive!) minimum cut algorithm to run on $H$.

We briefly remark on the choice of known upper bound $U$. If the edge weights are polynomially bounded by the number of vertices, $n$, a simple upper bound of the sum of weights of edges attached to any single vertex will do. If we do not consider this guarantee, Karger shows~\cite{Karger99} that the minimum weight edge $w$ in a maximum spanning tree has the property that the minimum cut must have weight between $w$ and $n^2w$. Thus, setting $U=n^2w$ gives only $O(\log n)$ values of $c'$ to try regardless of edge weights. The choice of an upper bound $U$ is further discussed in~\cite{Chekuri97}.

We now return to the issue of real-value weights in Lemma~\ref{sample}. This was described as a complication in~\cite{Chekuri97}, to which they substituted a heuristic method in order to achieve practicality. The approach we have described thus far is amenable to small constant-factor approximations. Suppose we replace $G$ with a graph $G'$ such that each edge weight is first normalized so the smallest weight edge has weight $1$, then all edge weights are multiplied by $100$ and rounded to the nearest integer. Normalizing has no effect on the relative sizes of cuts in $G'$. Rounding to the nearest integer when the smallest weight edge has weight at least $100$ has the effect that a cut of weight $x$ will take on a new weight in range $[.995x, 1.005x]$. Then the original minimum cut of $G$ corresponds to an at most $201/199$-times minimum cut of $G'$. Now, $G'$ can be represented as an unweighted multigraph and then sampled according to Lemma~\ref{sample}. In the resulting graph $H$, the minimum cut of $G$ corresponds to an at most $201/199 \cdot 7/6$-times minimum cut of $H$ with the choice $\epsilon = 1/6$. By adjusting constants throughout the rest of our approach, this shows we can treat real weighted graphs $G$ correctly. The other issue is how to do so efficiently.

If we consider $G'$ as an unweighted multigraph, the number of edges of $G'$ is proportional to the weight of edges of $G$, which may be quite large. However, we may also consider $G'$ as an integer-weighted graph, in which case we can sample each edge of $G'$ by drawing from the binomial distribution with probability $p$ and number of trials the weight of the respective edge. There are many methods to sample from the binomial distribution. One simple method that can be made efficient for our purposes is inverse transform sampling. Let $X$ denote a random variable sampled from the binomial distribution as described. In inverse transform sampling, we draw a number $u$ uniformly at random between $0$ and $1$, and then choose our sample $x$ to be the largest such that $P(X < x) \leq u$. Instead of having to sample a number of times equal to the weight of an edge, we must only compute the probabilities of the cumulative distribution function for the binomial distribution for all possible values that may result in $H$. We can make this efficient with the following observation. Say the weight of the minimum cut in $H$ is $\hat{c}$. Then a tree packing of $H$ has value at most $\hat{c}$, and in particular for a given edge, any weight beyond $\hat{c}$ is excess capacity that cannot be used in the tree packing. It follows that capping the weight of any edge of $H$ to the maximum size of the minimum cut in $H$, thus $O(\log n)$, will have no impact on the packing found. Thus, we must only compute $O(\log n)$ probabilities of the binomial distribution per edge, which can be done in total $O(\log n)$ time per edge.

The final choice is to pick a tree packing algorithm. Karger gives two options. The first is an algorithm by Gabow~\cite{Gabow91}, which computes a $c/2$ packing. The second is a more general approach by Plotkin-Shmoys-Tardos~\cite{Plotkin92}, which can find a packing a factor $(1+\epsilon')$ from the maximum packing, which has value in $[c/2, c]$. Karger describes the latter approach as simpler, using only minimum spanning tree computations. Although the paper~\cite{Plotkin92} does not explicitly give a routine for packing spanning trees, such a procedure is explicitly given in Thorup and Karger~\cite{Thorup00}, with credit given to Plotkin-Shmoys-Tardos~\cite{Plotkin92} and Young~\cite{Young95}. This procedure also appears in Gawrychowski et al.~\cite{Gawrychowski19}. We give the procedure in Algorithm~\ref{packpart} and state a version of Algorithm~\ref{packpart} with general epsilon in Algorithm~\ref{packpart2}.

\begin{algorithm}[H]
\caption{Obtain a Packing of Weight at least $(1-\epsilon)c/2$ from a Graph $G$}
\label{packpart2}
\begin{flushleft}
Let $G$ be a graph with $m$ edges and $n$ vertices.
\begin{enumerate}
  \item Initialize $\ell(e) \leftarrow 0$ for all edges $e$ of $G$. Initialize multiset $P \leftarrow \emptyset$. Initialize $W \leftarrow 0$.
  \item Repeat the following:
  \begin{enumerate}
    \item Find a minimum spanning tree $T$ with respect to $\ell(\cdot)$.
    \item Set $\ell(e) \leftarrow \ell(e) + \epsilon^2/(3 \ln m)$ for all $e \in T$. If $\ell(e) > 1$, return $W, P$.
    \item Set $W \leftarrow W + \epsilon^2/(3 \ln m)$.
    \item Add $T$ to $P$.
  \end{enumerate}
\end{enumerate}
\end{flushleft}
\end{algorithm}

We now give the general form of Lemma~\ref{packpartlemma} with proof.

\begin{lemma}[\cite{Plotkin92,Thorup00,Young95}]
\label{packpartlemma2}
Given $0 < \epsilon < 1$ and an undirected unweighted graph $G$ with $m$ edges, $n$ vertices, and minimum cut $c$, Algorithm~\ref{packpart2} returns a weighted packing of weight at least $(1-\epsilon)c/2$ in $O(mc \log n)$ time.
\end{lemma}

\begin{proof}
On each iteration, the weight of some tree is increased by $\epsilon^2 /(3 \ln m)$. Since the weight of the resulting packing is bounded by $c$, there are at most $3 c \ln m/\epsilon^2 = O(c \log n)$ iterations. The bottleneck in each iteration is the time to compute a minimum spanning tree in $G$. With an $O(m)$ time minimum spanning tree algorithm~\cite{Karger95c} our final time complexity is $O(m c \log n)$; an alternative way to achieve this runtime when Algorithm~\ref{packpart2} is used in Algorithm~\ref{packing2} was shown in Section~\ref{karger}. Correctness is given via Thorup and Karger~\cite{Thorup00}, Young~\cite{Young95}, and Plotkin-Shmoys-Tardos~\cite{Plotkin92}.
\end{proof}

Our full procedure for obtaining $\Theta(\log n)$ spanning trees for the rest of the algorithm is given in Algorithm~\ref{packing}. We give a version of Algorithm~\ref{packing} with general epsilons in Algorithm~\ref{packing2}.

\begin{algorithm}
\caption{Obtain $\Theta(\log n)$ Spanning Trees for the $2$-respect Algorithm}
\label{packing2}
\begin{flushleft}
Let $d$ denote the exponent in the probability of success $1-1/n^d$. Let $\epsilon_1, \epsilon_2, \epsilon_3 > 0$ be constants of approximation such that $f = 3/2-(\frac{2+\epsilon_1}{2-\epsilon_1})(1+\epsilon_2)(1-\epsilon_3)^{-1} > 0$ and $(1+\epsilon_2)^{-1}(1-\epsilon_3) > 2/3$. Let $b = 3(d+2) \ln n/ {\epsilon_2}^2$.
\begin{enumerate}
  \item Form graph $G'$ from $G$ by first normalizing the edge weights of $G$ so the smallest non-zero edge weight has weight $1$, then multiplying each edge weight by $\epsilon_1^{-1}$ and rounding to the nearest integer. Let $U$ be an upper bound for the size of the minimum cut of $G'$.
  \item Initialize $c' \leftarrow U$. Repeat the following:
  \begin{enumerate}
    \item Construct $H$ in the following way: for each edge $e$ of $G'$, let $e$ have weight in $H$ drawn from the binomial distribution with probability $p = \min(b/c', 1)$ and number of trials the weight of $e$ in $G'$. Cap the weight of any edge in $H$ to at most $\lceil (1+\epsilon_2) 12b \rceil$.
    \item Run Algorithm~\ref{packpart2} on $H$ with approximation $\epsilon_3$, considering an edge of weight $w$ as $w$ parallel edges. There are three cases:
    \begin{enumerate}
    	\item If $p = 1$, set $P$ to the packing returned and skip to step 3.
    	\item If the returned packing is of weight $\frac{1}{2}(1-\epsilon_3) (1+\epsilon_2)^{-1} b$ or greater, set $c' \leftarrow c'/6$ and repeat steps 2a and 2b, setting $P$ to the packing returned and then proceeding to step 3.
    	\item Otherwise, repeat steps 2a and 2b with $c' \leftarrow c'/2$.
    \end{enumerate}
  \end{enumerate}
	\item Return $\lceil -d\ln n/\ln (1-f)\rceil$ trees sampled uniformly at random proportional to their weights from $P$.
\end{enumerate}
\end{flushleft}
\end{algorithm}

We give the generalization of Lemma~\ref{packinglemma} for Algorithm~\ref{packing2} below.

\begin{lemma}
	\label{packinglemma2}
	Algorithm~\ref{packing2} returns a collection of $\Theta(\log n)$ spanning trees of $G$ in time $O(m \log^3 n)$ such that the minimum cut of $G$ $2$-respects at least one tree in the collection with high probability.
\end{lemma}

\begin{proof}
We first prove correctness. Consider general epsilons $\epsilon_1, \epsilon_2, \epsilon_3 > 0$, where in Algorithm~\ref{packing}, $\epsilon_1 = 1/100$ is the real-weight approximation, $\epsilon_2 = 1/6$ is the approximation for Lemmas~\ref{sample} and~\ref{overestimate}, and $\epsilon_3 = 1/5$ is the approximation for Algorithm~\ref{packpart} to return a packing of size $(1-\epsilon_3)c/2$ or greater.

Suppose for a particular $c'$ that $c' \geq 6c$, where $c$ is the size of the minimum cut in $G'$. Then by Lemma~\ref{overestimate}, $H$ will have minimum cut of size less than $b/3=(d+2)\ln n/{\epsilon_2}^2$ with high probability. A maximum tree packing of $H$ will have weight at most $\hat{c}$, the weight of the minimum cut in $H$, and thus the weight of the tree packing found by Algorithm~\ref{packpart2} will be at most $b/3 < \frac{1}{2}(1-\epsilon_3) (1+\epsilon_2)^{-1} b$ because $(1+\epsilon_2)^{-1}(1-\epsilon_3) > 2/3$. Therefore Algorithm~\ref{packing2} will proceed to the next iteration with $c' \leftarrow c'/2$. Note that the overall probability of failure from any of the $O(\log n)$ iterations of this step is at most $O(\log n \cdot n^{-(d+2)}) \leq n^{-d}$ for sufficiently large $n$.

Now suppose Algorithm~\ref{packpart2} returns a tree packing of weight $\frac{1}{2}(1-\epsilon_3) (1+\epsilon_2)^{-1} b$ or greater. By the above, $c' < 6c$ with high probability. If $c' \leq c$, Lemma~\ref{sample} says that the weight of the minimum cut is at least $(1+\epsilon_2)^{-1} b$ with high probability, unless $p > 1$. In the latter case, this implies the weight of the minimum cut is $O(\log n)$ and there is no need to apply sampling to $G'$. Consider the former case. The tree packing is of weight at least $(1-\epsilon_3)$ times half the minimum cut. It follows that the tree packing will be of weight at least $\frac{1}{2}(1-\epsilon_3) (1+\epsilon_2)^{-1} b$. The consequence of this is that if a tree packing of this weight or greater is found in step 2b, in addition to the bound $c' < 6c$, we also know $c' > c/2$ with high probability, since whenever $c' \leq c$, Lemma~\ref{sample} says the packing will have weight at least $\frac{1}{2}(1-\epsilon_3) (1+\epsilon_2)^{-1} b$, and we decrease $c'$ by a factor of $2$ in each iteration. Therefore, if we set $c' \leftarrow c'/6$, then in the next iteration we will have $c/12 < c' < c$.

Now consider the next iteration when the tree packing is returned. In sampling $H$, we only preserve weights in $H$ up to $\lceil (1+\epsilon_2) \cdot 12b \rceil$. Since $c' > c/12$, the expected size of the minimum cut in $H$ is at most $12b = 12 \cdot 3(d+2)\ln n/{\epsilon_2}^2$. Thus, with high probability, by Lemma~\ref{sample}, the size of the minimum cut in $H$ is at most $(1+\epsilon_2) 12b$, and as explained previously, we can afford to remove the capacity of any edge beyond $(1+\epsilon_2) 12b$ without impacting the returned packing. Now by Lemma~\ref{cuts} with $\alpha \leq \frac{2+\epsilon_1}{2-\epsilon_1}(1+\epsilon_2)$ and $\beta \geq \frac{1}{2} (1-\epsilon_3)$, a fraction of at least $f = 3/2-(\frac{2+\epsilon_1}{2-\epsilon_1})(1+\epsilon_2)(1-\epsilon_3)^{-1}$ of the trees in the packing found will $2$-constrain the minimum cut of $G$. The probability that no tree in a sample of size $t$ $2$-constrains the minimum cut is $(1-f)^t$. Solving for $t$ in $(1-f)^t = n^{-d}$ yields $t = -d\ln n/\ln(1-f)$. Therefore with probability at least $1-1/n^d$, at least one tree in the returned sample will $2$-constrain the minimum cut.

Time complexity can be proven as follows. Sampling $H$ can be done in $O(m \log n)$ time, as explained previously. Algorithm~\ref{packpart2} runs in $O(m' \hat{c} \log^2 n)$ time using a textbook $O(m\log n)$ minimum spanning tree algorithm, where $\hat{c}$ is the value of the minimum cut in $H$ and $m'$ is the number of edges in $H$, where weighted edges are considered parallel unit weight edges. Due to the sampling procedure, $m' = O(m \log n)$. To reduce this complexity, we can either use a linear time minimum spanning tree algorithm~\cite{Karger95c} or the implementation trick given in Section~\ref{karger}. If we use the latter, we reduce the effective $m'$ needed in Algorithm~\ref{packpart2} to $O(m)$. Further, in expectation, the value of the minimum cut $\hat{c}$ of $H$ doubles in each iteration of Algorithm~\ref{packing2}. A high probability statement can be made via an argument similar to Lemma~\ref{overestimate}. Therefore the cost of running Algorithm~\ref{packpart2} doubles in each iteration, with the final cost being $O(m \log^3 n)$, since $\hat{c} = O(\log n)$ by Lemma~\ref{sample}. This is a geometric series, so the entire cost is $O(m \log^3 n)$, and so Algorithm~\ref{packing2} runs in $O(m \log^3 n)$ time with high probability.
\end{proof}

Since Algorithm~\ref{packpart2} returns $O(\log n)$ trees, we could avoid sampling trees from the weighted packing and instead return all of them. We keep the sampling in Algorithm~\ref{packing2} because, depending on the constants, sampling may require less trees. Further, the above version of Algorithm~\ref{packpart2} is more versatile in that the packing algorithm can be changed. Observe that the entire algorithm is still only correct with high probability, since we required sampling $G'$ to construct graph $H$. Finally, returning all trees from Algorithm~\ref{packpart2} does not actually allow us to relax $\epsilon_1$, $\epsilon_2$, or $\epsilon_3$. The condition $f = 3/2-(1+\epsilon_1)(1+\epsilon_2)(1-\epsilon_3)^{-1} > 0$ is satisfied for all values of $\alpha$ and $\beta$ that guarantee at least one tree in a weighted packing of weight $\beta c$ $2$-constrains a cut of weight $\alpha c$ given by Lemma~\ref{cuts}.

Algorithm~\ref{packing2} is slightly different than the approach taken by Karger~\cite{Karger00}. In particular, Karger sparsifies edges of $H$ to have $m' = O(n \log n)$ and replaces an $O(m \log n)$ time minimum spanning tree computation in the tree packing algorithm with an $O(m)$ one, avoiding the implementation trick of Gawrychowski et al.~\cite{Gawrychowski19}. This gives complexity $O(n \log^3 n)$ for finding the $\Theta(\log n)$ spanning trees. However, since the remaining part of the algorithm also takes $O(m \log^3 n)$ time, we avoid these optimizations to simplify our procedures.

\end{document}